\documentclass[10pt,conference]{IEEEtran}
\IEEEoverridecommandlockouts
\usepackage{cite}
\usepackage{subfigure}
\usepackage{amsmath,amssymb,amsfonts,stmaryrd}
\usepackage{amsthm,mathtools}
\usepackage{algorithmic}
\usepackage{graphicx}
\usepackage{textcomp}
\usepackage{xcolor}
\usepackage{xparse} \DeclarePairedDelimiterX{\Iintv}[1]{\llbracket}{\rrbracket}{\iintvargs{#1}}
\NewDocumentCommand{\iintvargs}{>{\SplitArgument{1}{,}}m}
{\iintvargsaux#1} %
\NewDocumentCommand{\iintvargsaux}{mm} {#1\mkern1.5mu..\mkern1.5mu#2}
\newtheorem{remark}{Remark}

\newtheorem{theorem}{Theorem}
\newtheorem{lemma}{Lemma}

\newtheorem{assumption}{Assumption}

\newtheorem{definition}{Definition}
\def\BibTeX{{\rm B\kern-.05em{\sc i\kern-.025em b}\kern-.08em
    T\kern-.1667em\lower.7ex\hbox{E}\kern-.125emX}}

\begin{document}
\title{Affine Frequency Division Multiplexing For Communications on Sparse Time-Varying Channels}

\author{
Wissal Benzine\textit{$^1$}, Ali Bemani\textit{$^1$}, Nassar Ksairi\textit{$^1$}, and Dirk Slock\textit{$^2$} \\
\textit{$^1$}Mathematical and Algorithmic Sciences Lab, Huawei France R\&D, Paris, France \\ \textit{$^2$}Communication Systems Department, EURECOM, Sophia Antipolis, France\\
Emails:
\{wissal.benzine1, ali.bemani, nassar.ksairi\}@huawei.com, Dirk.Slock@eurecom.fr
}
\maketitle
\begin{abstract}
 This paper addresses channel estimation for linear time-varying (LTV) wireless propagation links under the assumption of \emph{double sparsity} i.e., sparsity in both the delay and the Doppler domains. Affine frequency division multiplexing (AFDM), a recently proposed waveform, is shown to be optimal (in terms of pilot overhead) for this problem. With both mathematical analysis and numerical results, the minimal pilot and guard overhead needed for achieving a target mean squared error (MSE) while performing channel estimation is shown to be the smallest when AFDM is employed instead of both conventional and recently proposed waveforms.
\end{abstract}


\section{Introduction}


Sparsity is an important feature of wireless propagation channels that can be exploited to improve channel estimation performance and/or pilot overhead. In wireless communications systems operating in sub-6GHz frequency bands, sparsity is mostly associated with the delay domain and manifests itself, as in \cite{Non_gaussian_MP,mmWave}, with a number of significant channel delay taps that is much smaller than the maximum delay spread.
In high-mobility scenarios, such as communications to/from high-speed trains (HST) or fast-moving cars, sparsity extends also to the Doppler domain. In high-frequency bands, Doppler domain sparsity appears even at moderate transmitter-receiver relative velocity values.

Delay-Doppler sparsity is assumed in \cite{mdpss}, and this sparsity is leveraged by performing channel estimation using compressed sensing methods. However, no study of the impact of sparsity on pilot overhead or on channel estimation performance is provided. In \cite{ofdm_dD_sparsity}, a special case of delay-Doppler sparsity is adopted in which one Doppler shift per delay tap is assumed. This model, which is restrictive for real-world wireless propagation channels, is then used to derive a lower bound on the number of pilots needed for guaranteed sparse recovery.
In \cite{otfs_ofdm_sparsity}, a  similarly restrictive delay-Doppler sparsity model with one Doppler frequency shift per delay tap was assumed when comparing orthogonal frequency division multiplexing (OFDM) and orthogonal time frequency space (OTFS) in terms of the \emph{pragmatic capacity} i.e., the
mutual information of the virtual channel having at its
input the constellation symbols \emph{excluding the pilot and guard symbols} and at its output the detector soft-outputs. 
While this overhead-aware comparison constitutes a step forward, the restrictive sparsity model does not allow to do the comparison under realistic propagation conditions nor to devise pilot patterns with adjustable time and frequency densities for different delay-Doppler sparsity levels. Nonetheless, such works point towards the fact that some waveforms are more suited to take advantage of delay-Doppler sparsity than others. For instance, channel estimation overhead in OTFS cannot be significantly reduced when the channel exhibits more sparsity unless non-orthogonal pilot-data multiplexing is used as in \cite{SP_otfs}. For such a scheme, sparsity in the channel delay-Doppler response lessens inter-pilot and pilot-data interference. However, the use of iterative detection methods becomes necessary which not only makes this approach require high computational complexity but also makes it prone to error propagation.
We thus restrict our work to the case of orthogonal resources for pilot and data symbols. In that context, we compare channel estimation performance and pilot overhead requirements of affine frequency division multiplexing (AFDM) \cite{BemaniAFDM_TWC, bemani2021afdm_ICC}, a newly proposed waveform based on the discrete affine Fourier transform (DAFT), to those of OFDM, OTFS and single-carrier modulation (SCM).
\subsection*{Contributions}
i) A statistical definition of delay-Doppler sparsity is provided. It is sufficiently general to cover at least three relevant types of delay-Doppler sparsity profiles. ii) This definition is used to get the statistical properties of the minimal pilot overhead needed for channel estimation. iii) Using these properties, closed-form asymptotic results for the average minimal pilot overhead of different waveforms are derived in the limit of a large communications frame size showing AFDM superiority.
\subsection*{Notations}
$\mathrm{Bernoulli}(p)$ is the Bernoulli distribution with probability $p$ and $\mathrm{B}(n,p)$ is the binomial distribution with parameters $(n,p)$.
If $\mathcal{A}$ is a set, $|\mathcal{A}|$ stands for its cardinality. For any real number $x$, $(x)_+$ stands for $\mathrm{max}(x,0)$. The set of all integers between $l$ and $m$ (including $l$ and $m$, $(l,m)\in\mathbb{Z}^2$) is denoted $\Iintv{l,m}$. For a matrix $\mathbf{M}$, $[\mathbf{M}]_{c}$ stands for the $c$-th column. The ceiling operation is denoted as $\lceil.\rceil$ and the floor operation as $\lfloor.\rfloor$. The modulo $N$ operation is denoted as $(\cdot)_N$.  
\section{Background: AFDM}
In AFDM, modulation is achieved through the use of DAFT. DAFT is a discretized version \cite{erseghe2005multicarrier} of the affine Fourier transform (AFT) \cite{healy2015linear,BemaniAFDM_TWC} with chirp $e^{-\imath2\pi (c_2k^2+{\frac{1}{N} }kn+c_1n^2)}$ as its kernel, where $c_1$ and $c_2$ are parameters that we adjust depending on the delay-Doppler characteristics of the channel.
Consider a set of quadrature amplitude modulation (QAM) symbols denoted $\{x_k\}_{k=0\cdots N-1}$. AFDM employs inverse DAFT (IDAFT) to map $\{x_k\}_{k=0\cdots N-1}$ to $\{s_n\}_{n=0\cdots N-1}$ as follows:
\begin{equation}
\label{eq:AFDM_mod}
    s_n = \frac{1}{\sqrt{N}}\sum_{k = 0}^{N-1}x_ke^{\imath2\pi (c_2k^2+{\frac{1}{N}}kn+c_1n^2)}, n= 0\cdots N-1
\end{equation}
with the following so called {\emph{chirp-periodic prefix}} (CPP)
\begin{equation}
    s_n = s_{N+n}e^{-\imath2\pi c_1(N^2+2Nn)},\quad n = -L_{\mathrm{CPP}}\cdots -1
\end{equation}
where $L_{\mathrm{CPP}}$ denotes an integer that is greater than or equal to the number of samples required to represent the maximum delay of the wireless channel. The CPP simplifies to a cyclic prefix (CP) whenever $2c_1N$ is integer and $N$ is even, an assumption that will be considered to hold from now on.

\section{System model}
 \subsection{Doubly sparse linear time-varying (DS-LTV) channels}
 Consider the following model of the variation with respect to the time index $n$ of the complex gain $h_{l,n}$ of the $l$-th path of a LTV channel with $L$ paths 
 \begin{equation}
     \label{eq:ch_model}
     h_{l,n}=\sum_{q=-Q}^{Q}\alpha_{l,q}e^{\imath 2\pi\frac{nq}{N}}, l=0\cdots L-1, n=0\cdots N-1
 \end{equation}
Note that this model is an on-grid approximation of a time-varying channel: the Doppler shifts are integer valued when normalized with the resolution associated with the transmission duration.
This model is just a first approximation used to make the presentation of the results and the mathematical proofs easier to follow. The case with off-grid Doppler frequency shifts will make the subject of a future work. The received samples after transmission over the channel are
\begin{equation}
    r_n = \textstyle\sum_{l = 0}^{L-1}s_{n-l}h_{l,n}\label{r_n} + w_n,\quad n= 0\cdots N-1.
\end{equation}
where $w_n\sim\mathcal{CN}\left(0,\sigma_w^2\right)$ represents the i.i.d. Gaussian noise process. After discarding the CPP and assuming that $L-1\leq L_{\rm CPP}$, the DAFT domain output symbols are
\begin{align}
    y_k &= \frac{1}{\sqrt{N}}\sum_{n = 0}^{N-1}r_ne^{-\imath2\pi (c_2k^2+{\frac{1}{N}}kn+c_1n^2)}, k= 0\cdots N-1\nonumber\\
    &=\sum_{l = 0}^{L-1}\sum_{q=-Q}^{Q}\alpha_{l,q} e^{\imath\frac{2\pi}{N}(Nc_1l^2-nl + Nc_2(n^2 - k^2))}x_n +\Tilde{{w}}_k,
    \label{eq:y_output_integer}
\end{align}
where the second equality is obtained using the input-output relation given in \cite{BemaniAFDM_TWC}, $\Tilde{{w}}_k$ is i.i.d. and $\sim\mathcal{CN}\left(0,\sigma_w^2\right)$
and where $ n = (k - q + 2Nc_1 l)_N$. Note how the Doppler components of delay taps are mixed in the DAFT domain in such a way that a path occupying the $(l,q)$ grid point in the delay-Doppler domain appears as a $q-2Nc_1l$ shift in the DAFT domain. We assume that the complex gains $\alpha_{l,q}$ follow a Bernoulli-Gaussian distribution \cite{sparse_globecom,Bernoulli_Gaussian} with respect to the hidden binary random variables $\left\{I_{l,q}\right\}_{l=0\cdots L-1,q=-Q\cdots Q}$ ($I_{l,q}$ takes the value one in the event that $\alpha_{l,q}\neq 0$ and the value zero otherwise) which 
are assumed to adhere to the following assumption. First define the events $\mathcal{I}_{l,q}\triangleq\{I_{l,q}=1\}$ and $\overline{\mathcal{I}}_{l,q}\triangleq\{I_{l,q}=0\}$ for $(l,q)\in\Iintv{0,L-1}\times\Iintv{-Q,Q}$.
\begin{assumption}
\label{assum:independence}
For any $(l_1,q_1),\ldots,(l_T,q_T)\in\Iintv{0,L-1}\times\Iintv{-Q,Q}$ ($T\leq\min(L,2Q+1)$) such that (s.t.) $l_s\neq l_t$ and $q_s\neq q_t$ for any distinct $s,t\in\Iintv{1,T}$, random variables $\{I_{l_t,q_t}\}_{t=1\cdots T}$ are mutually independent. Also, $\mathbb{P}\left[\mathcal{I}_{l,q}\right]=p_{\rm d}p_{\rm D}$ and $\mathbb{P}\left[\overline{\mathcal{I}}_{l,q}\right]=1-p_{\rm d}p_{\rm D}$ for all $(l,q)\in\Iintv{0,L-1}\times\Iintv{-Q,Q}$.
\end{assumption}
Define $\sigma_{\alpha}^2\triangleq\mathbb{E}\left[\left|\alpha_{l,q}\right|^2|I_{l,q}=1\right]$. The case with delay dependent conditional variances will be addressed in the future. Channel power normalization $\sum_{l=0}^{L-1}\sum_{q=-Q}^{Q}\mathbb{E}\left[\left|\alpha_{l,q}\right|^2\right]=1$ is obtained under Assumption \ref{assum:independence} by $\sigma_{\alpha}^2=\frac{1}{p_{\rm d} L p_{\rm D} (2Q+1)}$. In order to have more insight into the assumption, we give the following definitions of three types of delay-Doppler sparsity satisfying the assumption. The proof that these types satisfy indeed Assumption \ref{assum:independence} is given by Lemma \ref{lem:independence}.

\begin{definition}
\label{Def:dD_sparsity}
 [Figure \ref{fig:examples}-(a)] We say that a DS-LTV channel has a {\bf type-1 delay-Doppler sparsity} if there exist $0<p_{\rm d},p_{\rm D}<1$ and $L+2Q+1$ mutually independent random variables $\{\{I_l\}_{l=0\cdots L-1}, \{I_q\}_{q=-Q\cdots Q}\}$ such that $I_{l,q}=I_lI_q$, $I_l\sim\mathrm{Bernoulli}(p_{\rm d})$ and $I_q\sim\mathrm{Bernoulli}(p_{\rm D})$ for any $(l,q)$.
 \end{definition}
Note that according to Definition \ref{Def:dD_sparsity}, $\mathbb{E}[\sum_lI_l]=p_{\rm d}L$ is the mean number of active delay taps of the channel and can be thought of as the {\it delay domain sparsity level} while $\mathbb{E}[\sum_qI_q]=p_{\rm D}(2Q+1)$ is the mean number of active Doppler bins and is thus the {\it Doppler domain sparsity level}.
\begin{definition}
    \label{Def:dDl_sparsity}
 [Figure \ref{fig:examples}-(b)] We say that a DS-LTV channel has a {\bf type-2 delay-Doppler sparsity} if there exist $0<p_{\rm d},p_{\rm D}<1$ and $L(2Q+2)$  mutually independent random variables $\{I_l,\{I_q^{(l)}\}_{q=-Q\cdots Q}\}_{l=0\cdots L-1}$ such that $I_{l,q}=I_lI_q^{(l)}$, $I_l\sim\mathrm{Bernoulli}(p_{\rm d})$ and $I_q^{(l)}\sim\mathrm{Bernoulli}(p_{\rm D})$ for any $(l,q)$.
\end{definition}
Delay domain sparsity level under type-2 delay-Doppler sparsity is still equal to $p_{\rm d}L$ as under type-1 sparsity. However, $\mathbb{E}[\sum_qI_q^{(l)}]=p_{\rm D}(2Q+1)$ is now the Doppler  sparsity level {\it per delay bin} and not the ``total'' Doppler domain sparsity level. 
\begin{definition}
    \label{Def:dD_cluster_sparsity}
 [Figure \ref{fig:examples}-(c)] We say that a DS-LTV channel has a {\bf type-3 delay-Doppler sparsity} if there exist $0<p_{\rm d},p_{\rm D}<1$ and $L(2Q+2)$ random variables $\{I_l,\{I_q^{(l)}\}_{q=-Q\cdots Q}\}_{l=0\cdots L-1}$ such that $I_{l,q}=I_lI_q^{(l)}$ and any $I_{l_1}$ (resp. $I_q^{(l_1)}$) is independent from any $I_{l_2}$ (resp. $I_q^{(l_2)}$) for any $l_1\neq l_2$. Also, there exist an integer $1\leq R<Q$ such that $\frac{R}{2Q-R}=p_{\rm D}$ and $L$ mutually independent random variables $\left\{\Xi^{(l)}\right\}_{l=0\cdots L-1}$ each uniformly distributed on $\Iintv{-Q+\lfloor\frac{R-1}{2}\rfloor,Q-\lfloor\frac{R}{2}\rfloor}$ such that for any $l\in\Iintv{0,L-1}$, $\xi\in\Iintv{-Q+\lfloor\frac{R-1}{2}\rfloor,Q-\lfloor\frac{R}{2}\rfloor}$
 \begin{equation}
     \mathbb{P}\left[I_q^{(l)}=1|\Xi^{(l)}=\xi\right]=\left\{
     \begin{array}{cc}
          1,&  -\frac{R}{2}<q-\xi\leq \frac{R}{2}\\
          0,& \mathrm{otherwise}
     \end{array}
     \right.\:.
 \end{equation}
\end{definition}
Note that the above definition associates with each active delay tap a cluster of Doppler bins of cardinality $R$ and that random variable $\Xi^{(l)}$ represents the (random) value of the center frequency of the Doppler cluster of delay tap $l$.
\begin{figure}
  \centering
  \begin{tabular}{ c @{\hspace{5pt}} c @{\hspace{5pt}} c}
  \includegraphics[width=.3\columnwidth] {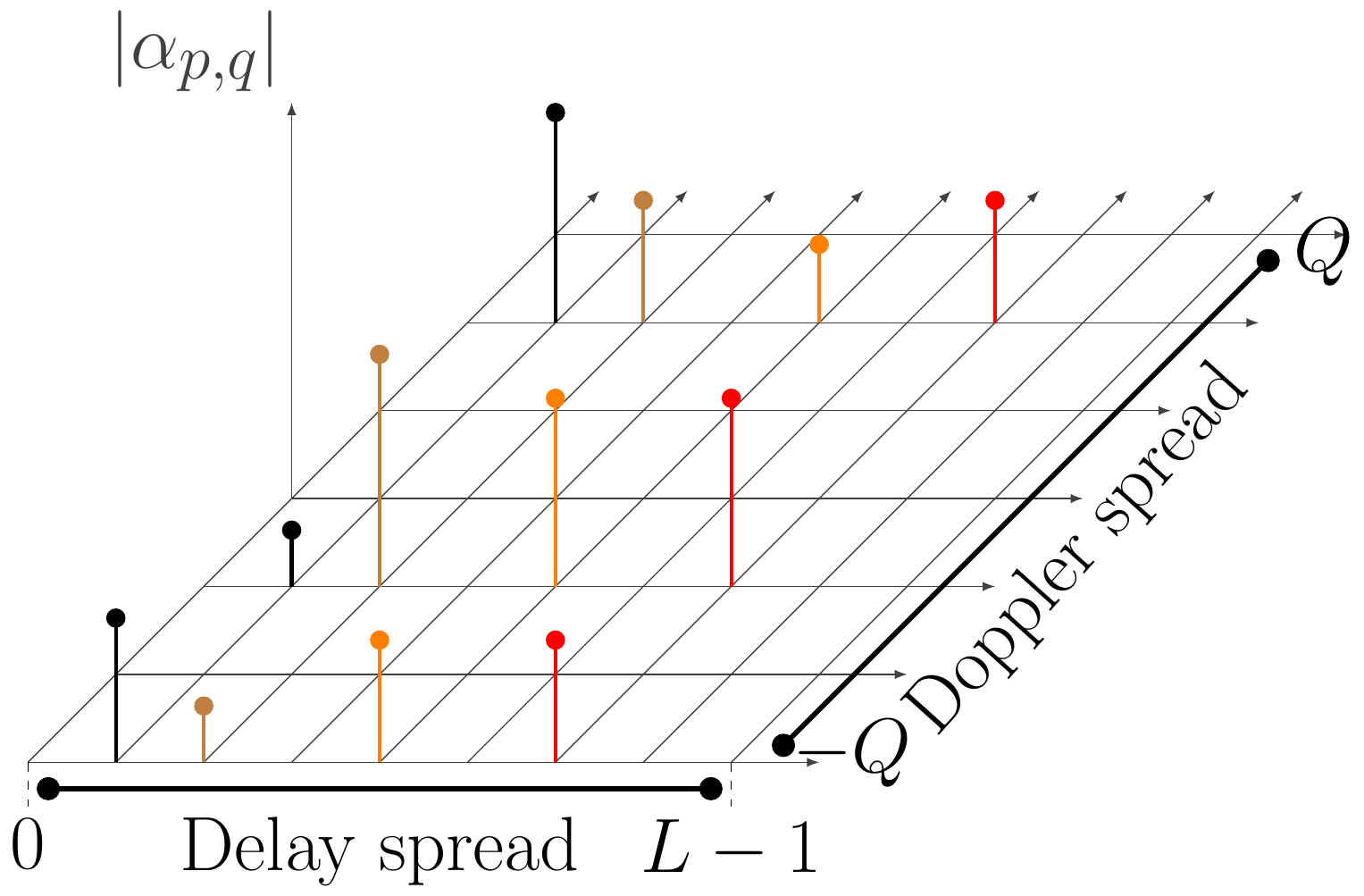} &
    \includegraphics[width=.3\columnwidth]{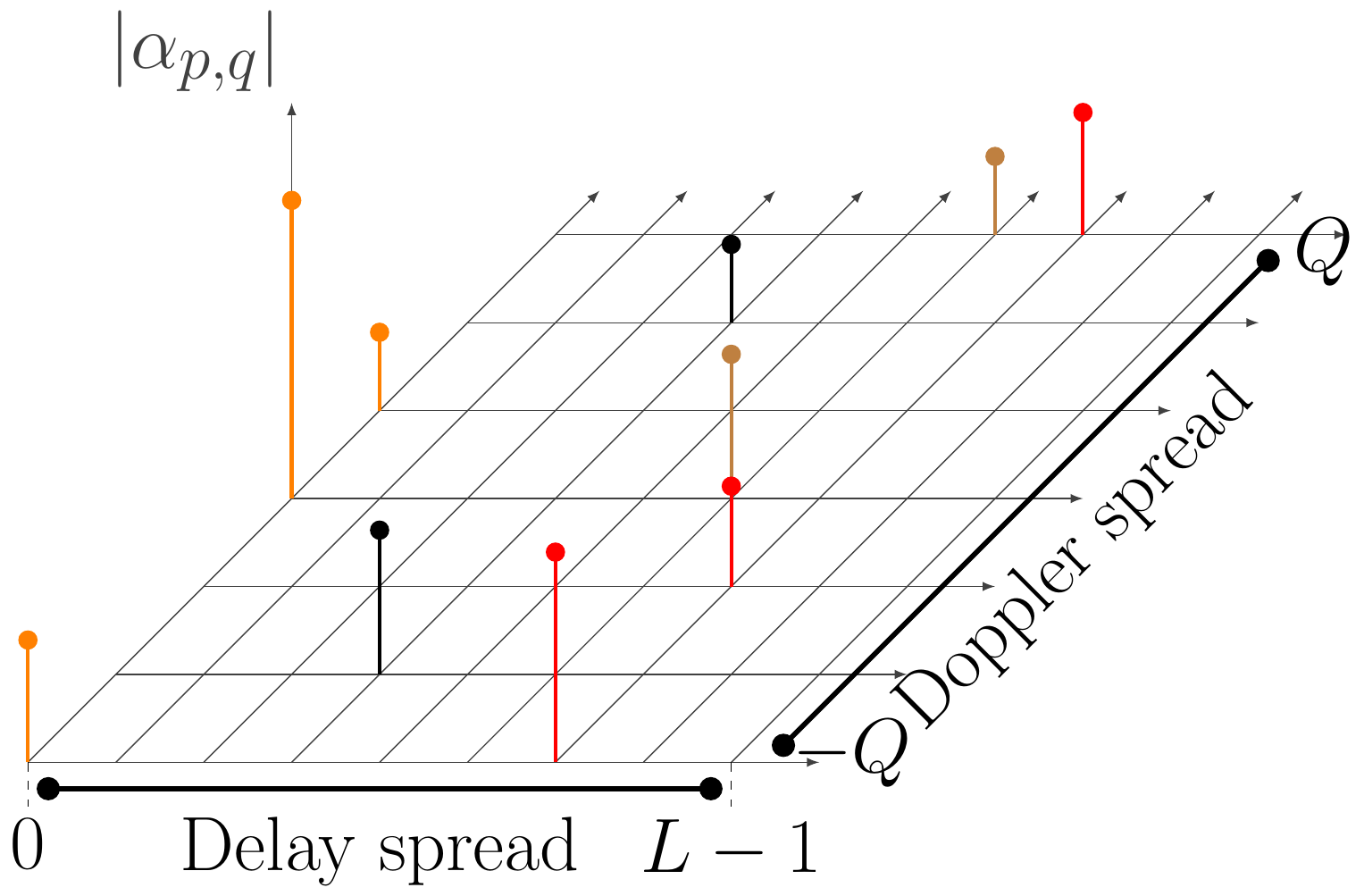} &
      \includegraphics[width=.3\columnwidth]{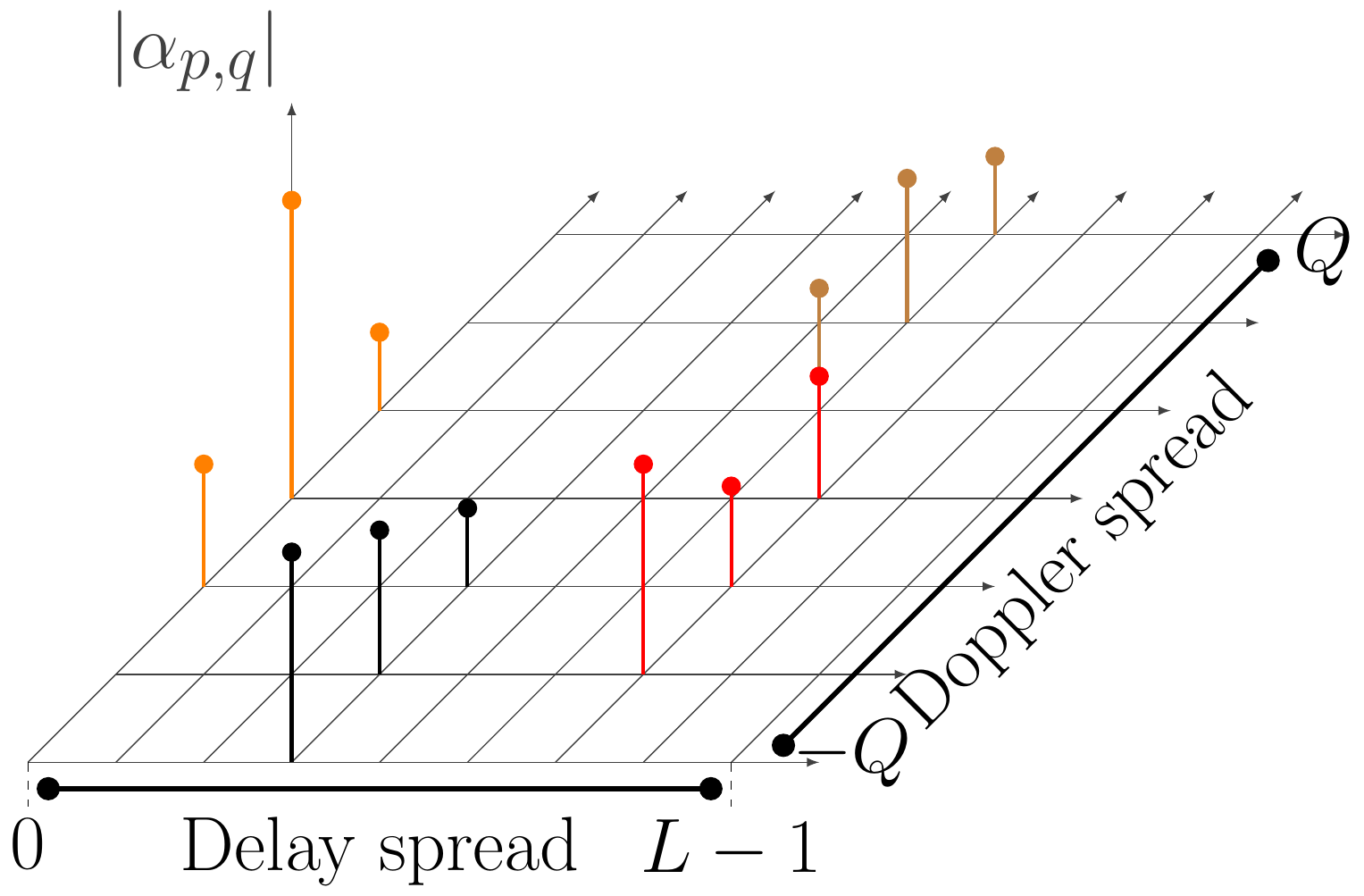} \\
    \small (a) &
      \small (b)&
      \small (c)
  \end{tabular}
  \medskip
  \caption{Examples of channels satisfying (a) Type-1 delay-Doppler sparsity, (b) Type-2 delay-Doppler sparsity, (c) Type-3 delay-Doppler sparsity}
  \label{fig:examples}
\end{figure}

\begin{lemma}
\label{lem:independence}
DS-LTV channels with type-1, type-2 or type-3 delay-Doppler sparsity satisfy Assumption \ref{assum:independence}.   
\end{lemma}
\begin{proof}
The proof of Lemma \ref{lem:independence} is provided in Appendix \ref{app:proof_lemma_assumption} in the case of type-3 delay-Doppler sparsity. The same arguments can be applied for the other two types of sparsity profiles. 
\end{proof}
The above models are not exhaustive. For instance, block sparsity can be extended to the delay domain.
Furthermore, each model can be extended by removing the on-grid approximation. In that case, $I_{l,q}$ will only represent the closest grid point in the delay-Doppler domain to a channel path instead of representing the path itself. Nonetheless, the three models capture important features of wireless channels in high frequency bands that are subject to user mobility.

 \subsection{DS-LTV channel estimation}
 Let $\boldsymbol{\alpha}\triangleq[\alpha_{l,q}]_{(l,q)\mathrm{s.t.}I_{l,q}=1}$ designate the vectorized form of the unknown channel gains associated with active delay-Doppler components. Let $\mathcal{P}\subset\Iintv{0,N-1}$ designate the indexes of the $M\left(2|c_1|N(L-1)+2Q+1\right)$ received samples associated with $M$ DAFT domain pilots (see Figure \ref{fig:pilot_pattern}) and define $\mathbf{y}_{\rm p}\triangleq[y_k]_{k\in\mathcal{P}}$ as the vectorized form of those received samples. Referring to \eqref{eq:y_output_integer}, we can write
 \begin{equation}
     \label{eq:yp}
     \mathbf{y}_{\rm p}=\underbrace{\mathbf{A}_{\cal P}\mathbf{M}\mathbf{A}_{\alpha}}_{\triangleq\mathbf{M}_{\rm p}}\boldsymbol{\alpha}+\mathbf{w}_{\rm p}
 \end{equation}
 where $[\mathbf{M}]_{l(2Q+1)+Q+q+1}=\boldsymbol{\Phi}\boldsymbol{\Delta}_q \boldsymbol{\Pi}^l \boldsymbol{\Phi}^H \mathbf{x}_{\rm p}$, $\mathbf{x}_{\rm p}$ is a $N$-long vector with entries equal to $p_1, \ldots, p_M$ (see Figure \ref{fig:pilot_pattern}) at the indexes of pilot symbols and to zero elsewhere, and $\mathbf{w}_{\rm p}\triangleq[\Tilde{{w}}]_{k\in\mathcal{P}}$. Here, $\mathbf{A}_{\cal P}$ is the $|\mathcal{P}|\times N$ matrix that chooses from a $N$-long vector the entries corresponding to $\mathcal{P}$, $\mathbf{A}_{\alpha}$ is the matrix that augments $\boldsymbol{\alpha}$ with zeros corresponding to $I_{l,q}=0$ resulting in a $L(2Q+1)$-long vector $\mathbf{A}_{\alpha}\boldsymbol{\alpha}$, $\boldsymbol{\Lambda}_{c}=\mathrm{diag}(e^{-\imath2\pi cn^2},n=0\cdots N-1)$, $\boldsymbol{\Delta}_q=\mathrm{diag}(e^{\imath2\pi qn},n=0\cdots N-1)$, $\boldsymbol{\Pi}$ is the $N$-order permutation matrix, $\boldsymbol{\Phi}=\pmb{\Lambda}_{c2} \mathbf{F}_N \pmb{\Lambda}_{c1}$ and $\mathbf{F}_N$ is the $N$-order discrete Fourier transform (DFT) matrix. The minimum mean squared error (MMSE) estimate, $\hat{\boldsymbol{\alpha}}$, of $\boldsymbol{\alpha}$ based on $\mathbf{y}_{\rm p}$ is given by \cite{kay2013fundamentals}
 \begin{equation}
     \label{eq:lmmse_alpha}
     \hat{\boldsymbol{\alpha}}=\sigma_{\alpha}^2(\sigma_{\alpha}^2\mathbf{M}_{\rm p}^{\rm H}\mathbf{M}_{\rm p}+\sigma_w^2\mathbf{I})^{-1}\mathbf{M}_{\rm p}^{\rm H}\mathbf{y}_{\rm p}\:.
 \end{equation}
 Note that the knowledge of the delay-Doppler profile (DDP) i.e., of $\left\{I_{l,q}\right\}_{l=0\cdots L-1,q=-Q\cdots Q}$, at the receiver side is here assumed. This bears similarities with the knowledge of the power delay profile (PDP) for linear time-invariant (LTI) channel estimation \cite{PDP}. The case where the DDP is not known will be addressed in future works using appropriate tools such as compressed sensing (CS) \cite{CS_for_CE,mdpss}.
 \begin{figure}
  \centering
  \includegraphics[scale=.45]{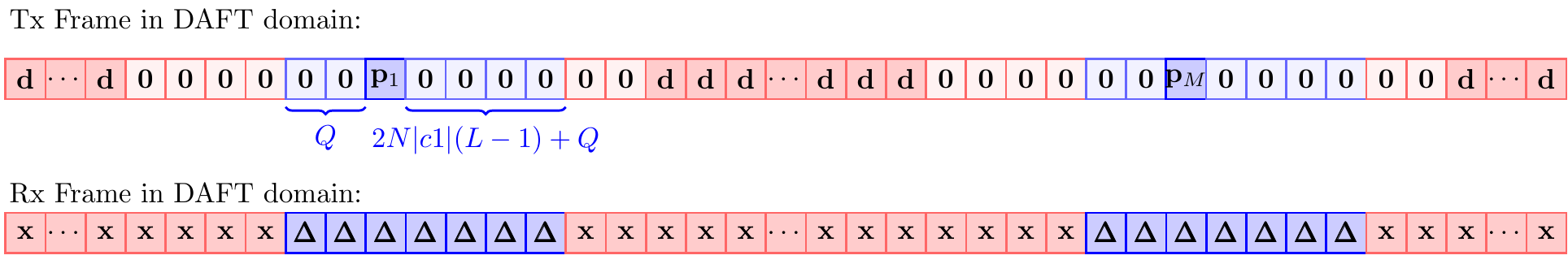}
  \caption{An AFDM symbol composed of data samples, $M$ pilot symbols and their guard samples. The received DAFT domain samples with indexes in $\mathcal{P}$ used for channel estimation are marked each with the symbol $\Delta$}
  \label{fig:pilot_pattern}
\end{figure}

For any $(l,q)$ satisfying $I_{l,q}=1$, define $\hat{\alpha}_{l,q}$ as the corresponding entry of vector $\hat{\boldsymbol{\alpha}}$. For any $(l,q)\in\Iintv{0,L-1}\times\Iintv{-Q,Q}$ such that $I_{l,q}=0$, set $\hat{\alpha}_{l,q}=0$. Finally, define
 \begin{equation}
    \label{eq:hat_h_def}
    \textstyle\hat{h}_{l,n}\triangleq\sum_{q=-Q}^{Q}\hat{\alpha}_{l,q}e^{\imath 2\pi\frac{nq}{N}},\quad n=0,\ldots, N-1\:.
 \end{equation}
as the resulting MMSE estimate of $h_{l,n}$. In what follows we give indications on how to set $M$ and $c_1$ based on the delay-Doppler sparsity level of the channel so that the minimal pilot overhead needed to guarantee a vanishing (with respect to an increasing signal-to-noise ratio (SNR)) mean squared error (MSE) $\mathbb{E}[\sum_{l=0}^{L-1}\frac{1}{N}\sum_{n=0}^{N-1}|h_{l,n}-\hat{h}_{l,n}|^2]=\mathbb{E}[\|\hat{\boldsymbol{\alpha}}-\boldsymbol{\alpha}\|^2]$.

\section{AFDM parameters setting for transmission over DS-LTV channels}
\label{sec:afdm_tuning}
Let $u_{m}^{(l,q)}$ ($m\in\mathbb{Z}$) be the individual DAFT domain impulse response of the part of the channel associated with delay-Doppler component $\alpha_{l,q}$. Since $\Iintv{0,L-1}\times\Iintv{-Q,Q}$ in the delay-Doppler domain maps to an interval in the DAFT domain that is either $\Iintv{-Q,2|c_1|N(L-1)+Q}$ if $c_1$ is negative or $\Iintv{-Q-2|c_1|N(L-1),Q}$ if $c_1$ is positive, the latter interval is the support of $u_{m}^{(l,q)}$. We designate by \emph{DAFT domain representation of the channel} the collection $\{u_{m}^{(l,q)}\}_{(l,q)\in\Iintv{0,L-1}\times\Iintv{-Q,Q}}$ of all individual DAFT domain impulse responses. Figure \ref{fig:DAFT}-(a) shows the DAFT domain representation of a channel in the case $c_1=\frac{-1}{2N}$ while Figure \ref{fig:DAFT}-(b) shows that representation when $c_1=\frac{-2}{2N}$.
\begin{figure}
  \centering
  \begin{tabular}{ c @{\hspace{2pt}} c }
  \includegraphics[width=.45\columnwidth]{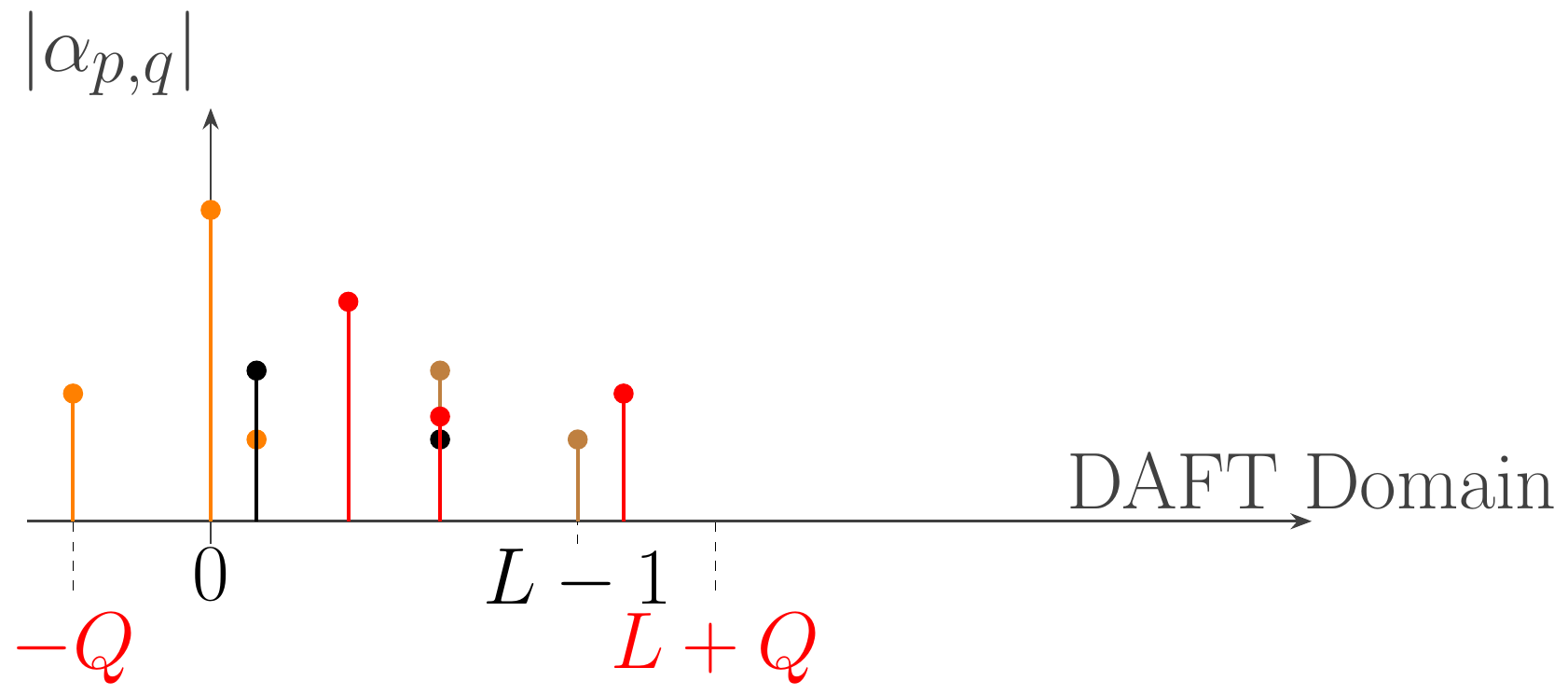} &
    \includegraphics[width=.45\columnwidth]{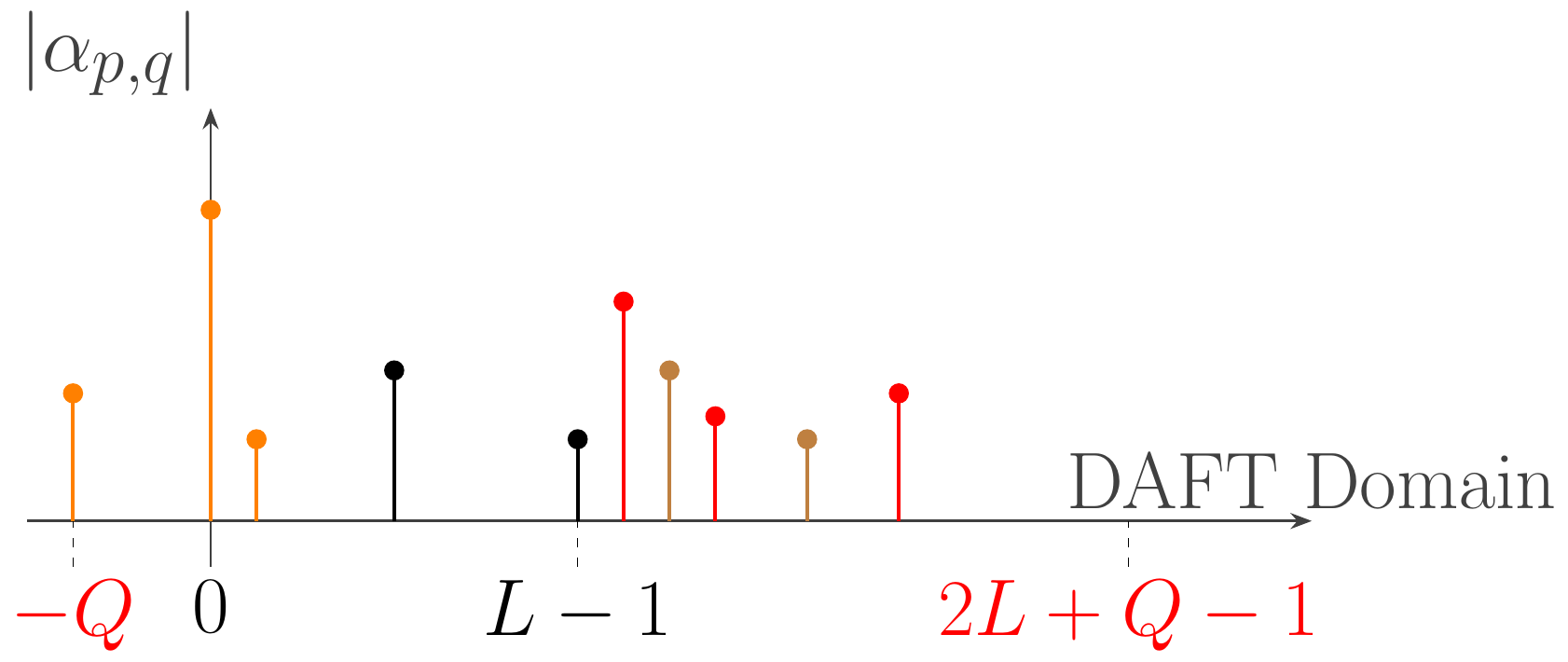} \\
    \small (a) &
      \small (b)
  \end{tabular}
  \medskip
  \caption{DAFT domain representation of the channel realization of Figure \ref{fig:examples}-(b) for different values of $c_1$, (a) $c_1=-1/2N$, (b) $c_1=-2/2N$}
  \label{fig:DAFT}
\end{figure}
In what follows we restrict $c_1$ to be negative without loss of generality. Define the random variable
$X_k\triangleq\left|\left\{(l,q)\in\Iintv{0,L-1}\times\Iintv{-Q,Q},I_{l,q}=1,q-2c_1Nl=k\right\}\right|$ for any $k\in\Iintv{-Q,2|c_1|N(L-1)+Q+1}$,  i.e., $X_k$ is the number of non-zero components $\alpha_{l,q}$ appearing at index $k$ in the DAFT domain representation. It is also the number of terms in the mixture of complex sinusoids that constitute the sample $y_k$ in \eqref{eq:y_output_integer} and is thus closely related to pilot overhead and channel estimation performance.  For instance, under a given channel realization, the minimal number of DAFT domain pilots needed for full identifiability i.e., for the measurement matrix $\mathbf{M}_{\rm p}$ in \eqref{eq:yp} to have full column rank, should be at least equal to ${\max}_{k\in\Iintv{-Q,2|c_1|N(L-1)+Q+1}}X_k$. We will also show that the distribution of $X_k$ affects directly the MSE of  $\hat{\boldsymbol{\alpha}}$.This is why we examine in what follows that probability distribution.
%
First, by referring to the signal relation in \eqref{eq:y_output_integer} we can see that with ``enough'' sparsity i.e., if the number of nonzero channel components is sufficiently smaller than the support $\Iintv{-Q,2|c_1|N(L-1)+Q+1}$ of the channel DAFT domain representation, it is unlikely that $X_k$ takes large values and hence it is unlikely that a large number of DAFT domain pilots would be needed to get a target estimation error performance. This can be seen from Figure \ref{fig:DAFT} where ${\max}_{k}X_k=3$ when $\left|c_1\right|=\frac{1}{2N}$ and ${\max}_{k}X_k=1$ when $\left|c_1\right|=\frac{2}{2N}$. The following lemma and the ensuing theorem give a rigorous confirmation of the above intuition.
\begin{lemma}
\label{lem:sub_binomial}
For $c_1=-\frac{P}{2N}$ ($P\in\mathbb{N}^*$) and any $k\in\Iintv{-Q,2|c_1|N(L-1)+Q}$ the complementary cumulative distribution function (CCDF) of $X_k$ under Assumption \ref{assum:independence} is upper-bounded by the CCDF of $\mathrm{B}((2\lceil\frac{Q}{P}\rceil+1,p_{\rm d} p_{\rm D})$.
\end{lemma}
\begin{proof}
The proof of Lemma \ref{lem:sub_binomial} is provided in Appendix \ref{app:proof_lemma_sub_nomial}. 
\end{proof}
For tractability and more insights, the asymptotic regime for $N,L,Q$ defined by the following assumption will be helpful. Note that the numerical results given in Section \ref{sec:simulations} are not asymptotic but are obtained with finite values of $N,L,Q$.
\begin{assumption}
    \label{assum:bigO}
    $L=O(K)$, $Q=O(K)$, $p_{\rm d}L=O\left(K^{\kappa_{\rm d}}\right)$ and $p_{\rm D}(2Q+1)=O\left(K^{\kappa_{\rm D}}\right)$ for some $\kappa_{\rm d},\kappa_{\rm D}\in[0,1)$.
\end{assumption}
\begin{remark}
    Assuming $L=O(K)$ and $Q=O(K)$ implies that $N=O\left(K^2\right)$. Indeed, assuming that the maximum delay $L$ increases to infinity as $K$ implies that the transmission bandwidth increases at the same rate. Also, assuming that the maximum Doppler shift $Q$ increases as $K$ implies that the transmission duration increases at the same rate. Therefore, the frame size in samples i.e., $N$, increases as $K^2$.
\end{remark}
\begin{theorem}
    \label{theo:adaptive_M}
    Under Assumption \ref{assum:bigO} and the conditions of Lemma \ref{lem:sub_binomial}, if we set $P$ s.t. $(L-1)P+2Q+1=O(p_{\rm d} L p_{\rm D} (2Q+1))$ then DAFT domain pilot overhead needed for the MSE $\mathbb{E}[\|\hat{\boldsymbol{\alpha}}-\boldsymbol{\alpha}\|^2]$ to tend to zero as $K\to\infty$ and $\sigma_w^2\to 0$ is $O(K^{\kappa_{\rm d}+\kappa_{\rm D}})$.
\end{theorem}
\begin{proof}
A sketch of the proof is given in Appendix \ref{app:proof_theorem_adaptive_M}.
\end{proof}
\begin{remark}
    Theorem \ref{theo:adaptive_M} implies that AFDM is order-optimal in terms of channel estimation overhead for DS-LTV channels since the total overhead needed for vanishing channel estimation MSE has the same asymptotic order as the smallest possible overhead which is equal to the average number of unknowns and thus to $\mathbb{E}\left[\left|\left\{(l,q)\in\Iintv{0,L-1}\times\Iintv{-Q,Q}\mathrm{s.t.} I_{l,q}=1\right\}\right|\right]=
        p_{\rm d}Lp_{\rm D}(2Q+1)=O\left(K^{\kappa_{\rm d}+\kappa_{\rm D}}\right)$.
    This optimality of AFDM is confirmed by the comparison done in the following section of its channel estimation overhead and channel estimation error performance to that required by OFDM, OTFS and SCM.
\end{remark}

\section{Numerical Results}
\label{sec:simulations}

For a $N$-long SCM transmission, estimating $\boldsymbol{\alpha}$ requires a minimal number $\check{M}_{\min}$ of time domain pilots dispersed throughout the frame and each with $2(L-1)$ guard samples \cite{scm_pilots} as shown in Figure \ref{fig:pilot_pattern_scm}. Using the same arguments as in the proof of Theorem \ref{theo:adaptive_M} gives $\mathbb{E}[\check{M}_{\min}]=p_{\rm D}(2Q+1)$ resulting in a total overhead of $2p_{\rm D}(2Q+1)(L-1)=O(K^{1+\kappa_{\rm D}})$.
\begin{figure}
  \centering
  \includegraphics[scale=.45]{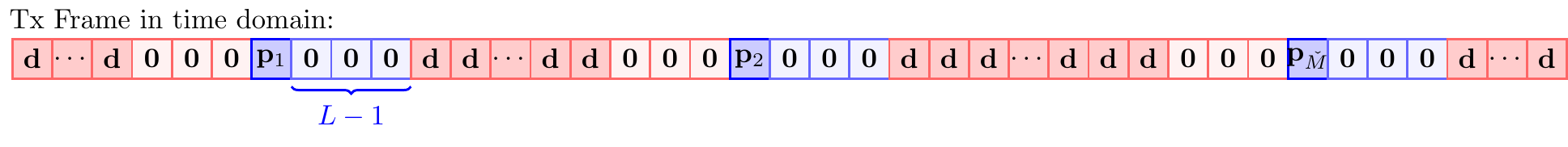}
  \caption{An example of a SCM frame composed of data samples and pilot symbols, each of the latter surrounded by $2L-1$ guard samples.}
  \label{fig:pilot_pattern_scm}
  \vspace{-2mm}
\end{figure}
As for the $N$-long OFDM frame of Figure \ref{fig:pilot_pattern_ofdm}, a minimal number $\check{M}_{\min}$ of OFDM symbols each costing $L-1$ in CP overhead needs each to contain a minimal number $\hat{M}_{\min}$ of pilot subcarriers each with $4Q_0$ guard subcarriers \cite{ofdm_pilots}. If the frame contains at least $2Q+1$ OFDM symbols, the Doppler shifts experienced within each of them is purely fractional and $Q_0$ is a constant that can be set to achieve a target low level of data-pilot interference ($Q_0=1$ in the figure). If $\hat{M}_{\min}(4Q_0+1)$ exceeds the size of the OFDM symbol, the whole symbol is made into pilots and there is no need for its length to exceed $L$. Again, the arguments of the proof of Theorem \ref{theo:adaptive_M} can be used to show that $\mathbb{E}[\hat{M}_{\min}]=p_{\rm d}L$ and $\mathbb{E}[\check{M}_{\min}]=p_{\rm D}(2Q+1)$. The total pilot overhead is thus $p_{\rm D}(2Q+1)(L-1+\min(L,p_{\rm d}L(4Q_0+1)))=O(K^{1+\kappa_{\rm D}})$.
\begin{figure}
  \centering
  \includegraphics[scale=.4]{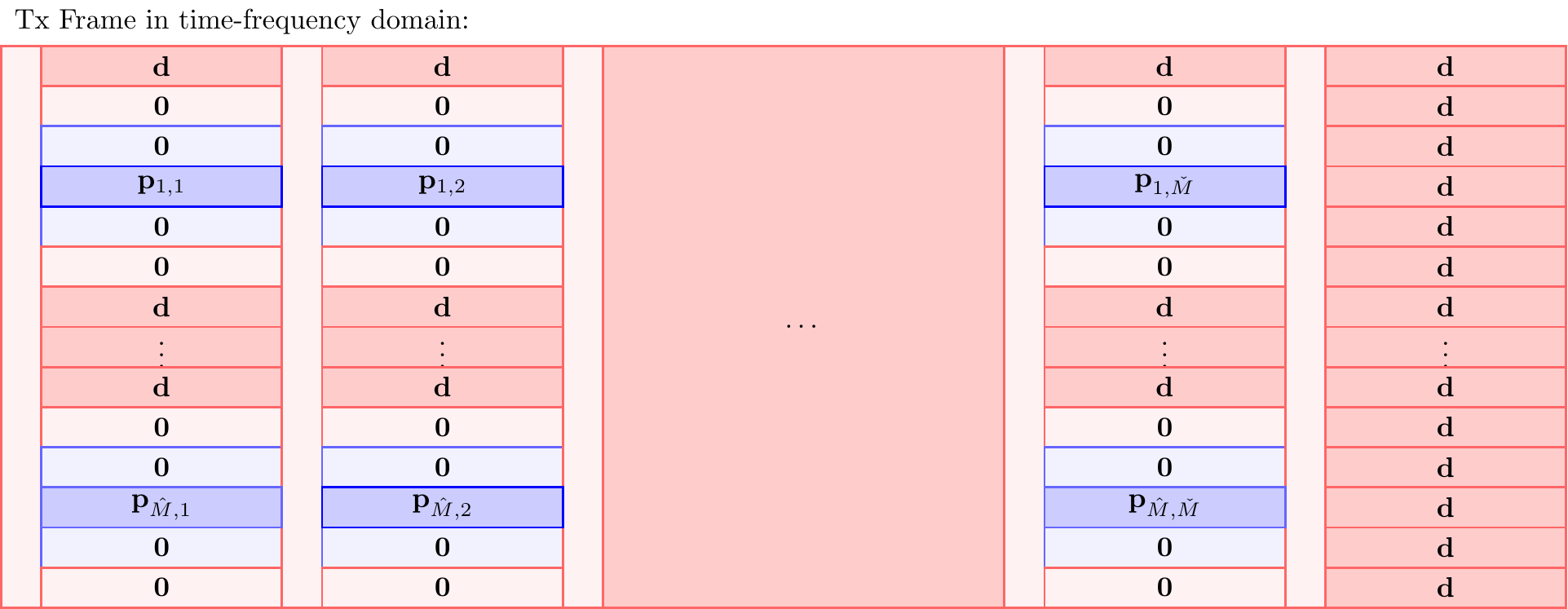}
  \caption{OFDM frame with pilot (blue), guard (light blue and red) and data (red) subcarriers. Each symbol is preceded by $L-1$ CP samples (light red)}
  \label{fig:pilot_pattern_ofdm}

\end{figure}
OTFS with orthogonal data-pilot resources \cite{embedded_otfs_bem} requires as shown in Figure \ref{fig:pilot_pattern_otfs} at least $\min(4Q+1,N_{\rm otfs})\times\min(2L-1,M_{\rm otfs})=O(K^2)$ pilot samples irrespective of sparsity level.
\begin{figure}
  \centering
  \includegraphics[scale=.4]{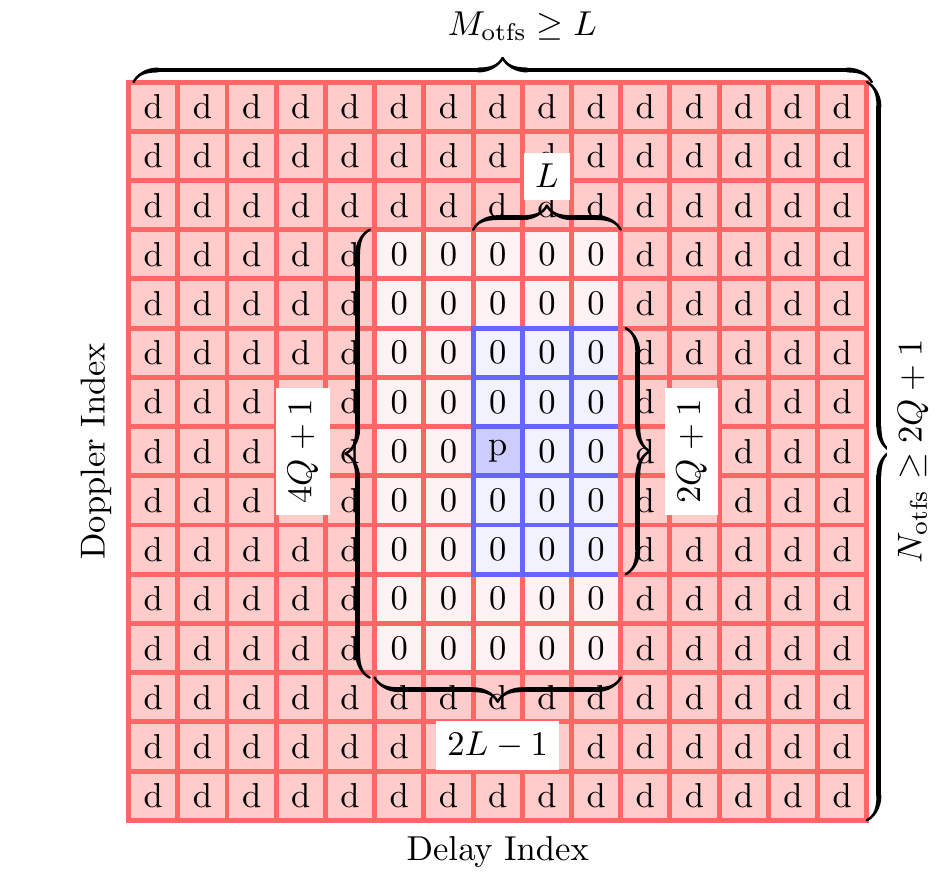}
  \caption{An OTFS symbol composed in the Zak domain of data samples (red), a pilot sample (blue) and guard samples (light blue and red)}
  \label{fig:pilot_pattern_otfs}
\end{figure}

In Figure \ref{fig:simulations1_2} channel estimation MSE of AFDM and OFDM are compared. We used 100 realizations of channels having a type-1 delay-Doppler sparsity for all simulations presented in this section with $N=8192, L=60, Q=15$ (corresponding to a 12 MHz transmission at a 70 GHz carrier frequency, a relative moving speed of 340 km/h and a delay spread of 5 $\mu\mathrm{s}$) and $p_{\rm d}=0.2$. In solid lines, the number $M$ of AFDM pilots and $\check{M}$ of OFDM pilots were set in each channel realization a certain amount above $M_{\min}$ and $\check{M}_{\min}$ respectively to get a $10^{-3}$ MSE at $\mathrm{SNR}=20$ dB. The dashed line is the MSE of OFDM with $\check{M}$ reduced to make pilot overhead equal to that of AFDM. As dictated by Theorem \ref{theo:adaptive_M}, In Figure \ref{fig:simulations1_2}-(a) where $p_{\rm D}=0.2$ AFDM with $P=1$ has the lowest overhead with $\mathbb{E}[M]= 7$ while it is AFDM with $P=2$ and $\mathbb{E}[M]= 7$ in Figure \ref{fig:simulations1_2}-(b) where $p_{\rm D}=0.3$.
  \begin{figure}
  \centering
  \begin{tabular}{ c  }
  \includegraphics[width=.7\columnwidth]{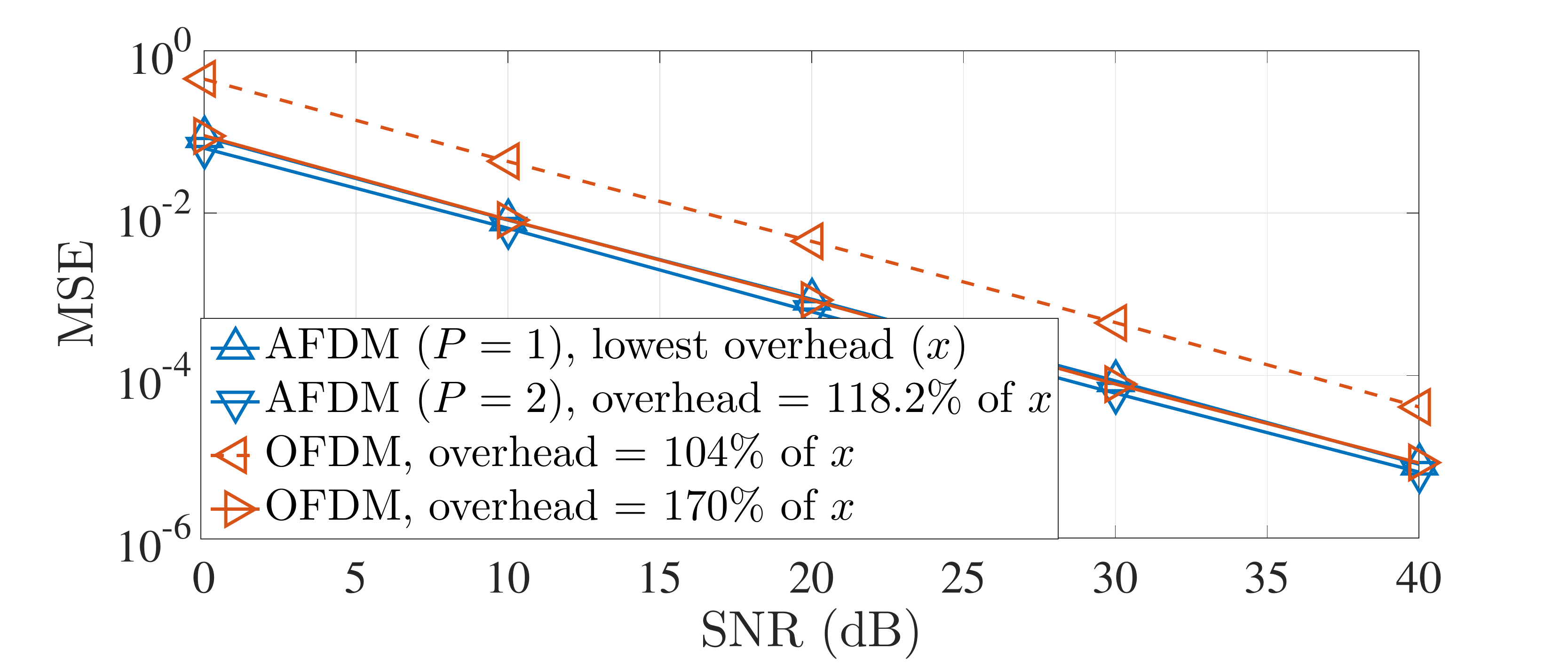} \\
  \small (a) \\
    \includegraphics[width=.7\columnwidth]{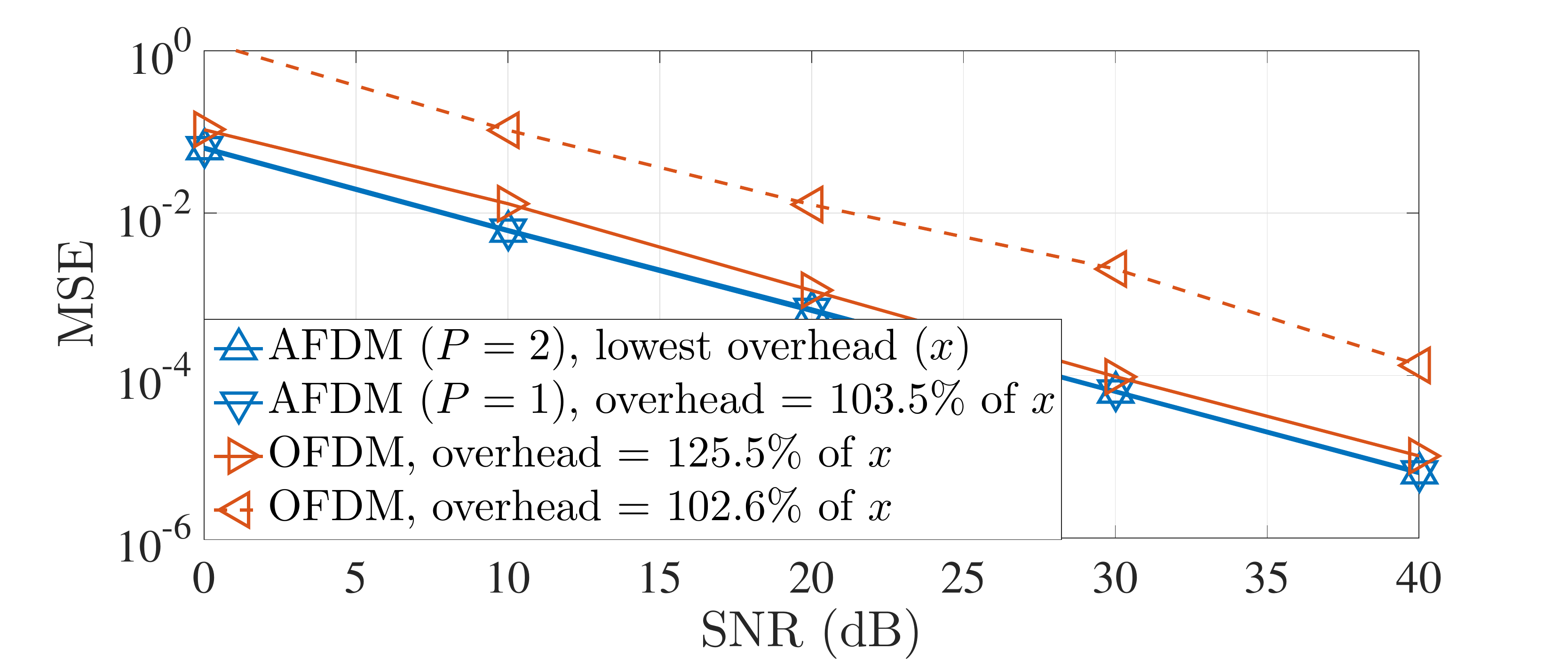} \\
      \small (b)
  \end{tabular}
  \medskip
  \caption{MSE performance for $N=8192, L=60, Q=15, p_{\rm}=0.2$, $M, \check{M}$ set as explained in the text and $\hat{M}=\hat{M}_{\min}$. (a) $p_{\rm D}=0.2$, (b) $p_{\rm D}=0.3$}
  \label{fig:simulations1_2}
\end{figure}

In Figure \ref{fig:overhead_vs_sparsity}, the average pilot overhead needed to achieve the target MSE is plotted for different values of $p_{\rm d}$ while $p_{\rm D}=0.2$.
\begin{figure}
    \centering
    \includegraphics[width=.7\columnwidth]{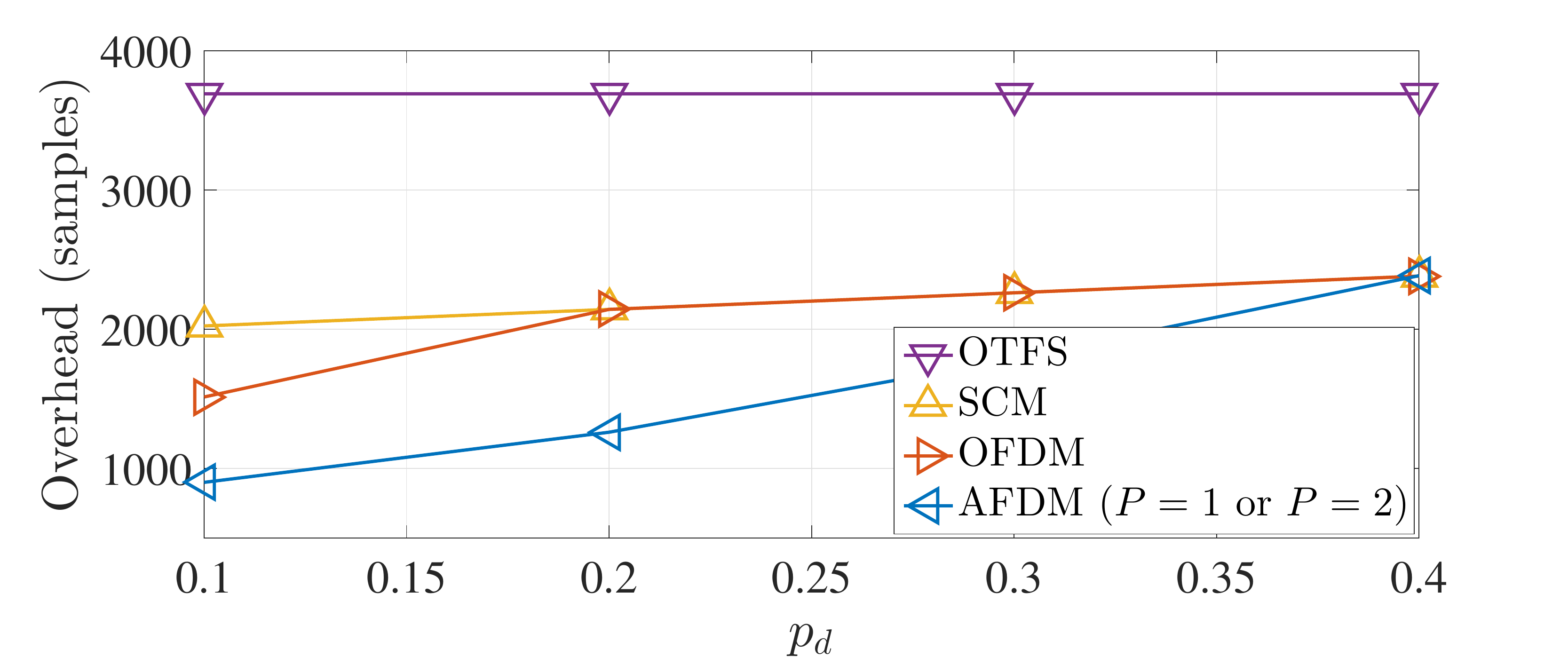}
    \caption{Channel estimation overhead for a target $\mathrm{MSE}=10^{-3}$ at $\mathrm{SNR}=20$ dB for $N=8192, L=60, Q=15, p_{\rm D}=0.2$}
    \label{fig:overhead_vs_sparsity}
\end{figure}
As expected, the gain with respect to OFDM, SCM and OTFS is the largest when sparsity is the highest. When there is no sparsity ($p_{\rm d}$ close to 1), performance measures other than pilot overhead can be used e.g., diversity order or channel delay-Doppler components separability. AFDM has been shown to achieve the optimal diversity order of LTV channels \cite{BemaniAFDM_TWC} in the general case irrespective of sparsity. 

\section{Conclusions}
Channel estimation for doubly dispersive wireless links that are sparse in both the delay and the Doppler domains was addressed. A special focus was given to the minimal pilot overhead required by different waveforms to achieve a target error performance while solving that problem. AFDM was shown to be optimal with respect to that performance measure when compared to SCM, OFDM and OTFS using both mathematical analysis and numerical results. Future work will address the problem without the on-grid approximation and provide numerical results and theoretical analysis for the problem in a compressed-sensing setting.

\appendices
\section{Proof of Lemma \ref{lem:independence}}
\label{app:proof_lemma_assumption}
\begin{proof}
We will only prove the independence of events $\mathcal{I}_{l_1,q_1}$, $\mathcal{I}_{l_2,q_2}$ for the case $T=2$. The independence of $\mathcal{I}_{l_1,q_1}$, $\overline{\mathcal{I}}_{l_2,q_2}$ and of $\overline{\mathcal{I}}_{l_1,q_1}$, $\overline{\mathcal{I}}_{l_2,q_2}$ and the result in the case $T>2$ follow similarly. Indeed, for a type-3 delay-Doppler sparsity
\begin{align}
\label{eq:lemma_proof_dDc}
    \mathbb{P}\left[\mathcal{I}_{l_1,q_1},\mathcal{I}_{l_2,q_2}\right]
    =&\mathbb{P}[I_{l_1}=1,I_{q_1}^{(l_1)}=1,I_{l_2}=1,I_{q_2}^{(l_2)}=1]\nonumber\\
    =&\mathbb{P}[I_{l_1}=1,-\frac{R}{2}<q_1-\Xi^{(l_1)}\leq\frac{R}{2}]\times\nonumber\\
    &\mathbb{P}[I_{l_2}=1,-\frac{R}{2}<q_2-\Xi^{(l_2)}\leq\frac{R}{2}]\nonumber\\
    =&\mathbb{P}[I_{l_1,q_1}=1]\mathbb{P}[I_{l_2,q_2}=1]\nonumber\\
    =&\mathbb{P}[\mathcal{I}_{l_1,q_1}]\mathbb{P}\left[\mathcal{I}_{l_1,q_1}\right]
\end{align}
This proves that type-3 channels satisfy the independence part of Assumption \ref{assum:independence}. Proving $\mathbb{P}[\mathcal{I}_{l,q}]=p_{\rm d}p_{\rm D}$ follows from the fact that $\mathbb{P}[\mathcal{I}_{l,q}]=\mathbb{P}[I_l=1]\mathbb{P}[I_q^{(l)}=1]$ with $\mathbb{P}[I_l=1]=p_{\rm d}$ due to Definition \ref{Def:dD_cluster_sparsity} and $\mathbb{P}[I_q^{(l)}=1]=\mathbb{P}[-\frac{R}{2}<q-\Xi^{(l)}\leq \frac{R}{2}]=\frac{R}{Q-\lfloor\frac{R}{2}\rfloor+Q-\lfloor\frac{R-1}{2}\rfloor+1}=p_{\rm D}$ where the second equality is due to the uniform distribution on $\Iintv{-Q+\lfloor\frac{R-1}{2}\rfloor,Q-\lfloor\frac{R}{2}\rfloor}$ of $\Xi^{(l)}$ and the third is due to the condition $\frac{R}{2Q-R}=p_{\rm D}$.
\end{proof}

\section{Proof of Lemma \ref{lem:sub_binomial}}
\label{app:proof_lemma_sub_nomial}
\begin{proof}
We only consider $P=1$. The proof for $P>1$ follows the same arguments. For any $k\in\Iintv{-Q,L-1+Q}$, define
\begin{equation}
\begin{split}
    \mathcal{Q}_k\triangleq\{l\in \Iintv{0,L-1} \mathrm{s.t.} \exists q\in\Iintv{-Q,Q}, q+l=k \}\\
    =\Iintv{k-Q,k+Q}\cap\Iintv{0,L-1}\triangleq \Iintv{l_{k,\max}, l_{k,\min}}
\end{split}
\end{equation}
For any $M\in\Iintv{0, 2Q+1}$ define
$\mathcal{L}_{k,M}\triangleq\left(
\begin{array}{c}
     \mathcal{Q}_k  \\
     M
\end{array}
\right)$ as the set of all $M$-size subsets of $\mathcal{Q}_k$. Then $\left|\mathcal{L}_{k,M}\right|=\left(\begin{array}{c}
     \left|\mathcal{Q}_k\right|  \\
     M
\end{array}\right)$ and
\begin{equation}\begin{split}
\mathbb{P}[X_k=M]=\\
\mathbb{P}[\bigcup_{\stackrel{(l_1,\ldots,l_M)}{\in\mathcal{L}_{k,M}}}\{\bigcap_{\stackrel{l\in}{\{l_1,\ldots,l_M\}}}\mathcal{I}_{l,l-k}\quad\cap\bigcap_{\stackrel{l\in\mathcal{Q}_k\setminus}{\{l_1,\ldots,l_M\}}}\overline{\mathcal{I}}_{l,l-k}\}]\\
=\sum_{\stackrel{(l_1,\ldots,l_M)}{\in\mathcal{L}_{k,M}}}\mathbb{P}[\bigcap_{\stackrel{l\in}{\{l_1,\ldots,l_M\}}}\mathcal{I}_{l,l-k}\quad\cap\bigcap_{\stackrel{l\in\mathcal{Q}_k\setminus}{\{l_1,\ldots,l_M\}}}\overline{\mathcal{I}}_{l,l-k}]\\
=\sum_{\stackrel{(l_1,\ldots,l_M)}{\in\mathcal{L}_{k,M}}}\prod_{\stackrel{l\in}{\{l_1,\ldots,l_M\}}}\mathbb{P}[\mathcal{I}_{l,l-k}]\prod_{\stackrel{l\in\mathcal{Q}_k\setminus}{\{l_1,\ldots,l_M\}}}\mathbb{P}[\overline{\mathcal{I}}_{l,l-k}]
\end{split}
\label{eq:Pr_Xk_M}
\end{equation}
where the second equality follows because the terms of the union are all disjoint events and where the third equality
is due to the independence property established by Assumption \ref{assum:independence} (in each term of the sum in the right-hand side of the second equality in \eqref{eq:Pr_Xk_M}, each pair of events is either $\left(\mathcal{I}_{l_1,q_1},\mathcal{I}_{l_2,q_2}\right)$, $\left(\mathcal{I}_{l_1,q_1},\overline{\mathcal{I}}_{l_2,q_2}\right)$ or $\left(\overline{\mathcal{I}}_{l_1,q_1},\overline{\mathcal{I}}_{l_2,q_2}\right)$ with $l_1\neq l_2$ and $q_1\neq q_2$).

If $k\in\Iintv{Q,L-1-Q}$, $l_{k,\min}=k-Q$ and $l_{k,\max}=k+Q$ i.e., $|\mathcal{Q}_k|=2Q+1$ and $\left|\mathcal{Q}_k\setminus\{l_1,\ldots,l_M\}\right|=2Q+1-M$ as shown in Figure \ref{fig:interval_k}-(a).
\begin{figure}
  \centering
  \begin{tabular}{ c @{\hspace{5pt}} c }
  \includegraphics[width=.4\columnwidth]{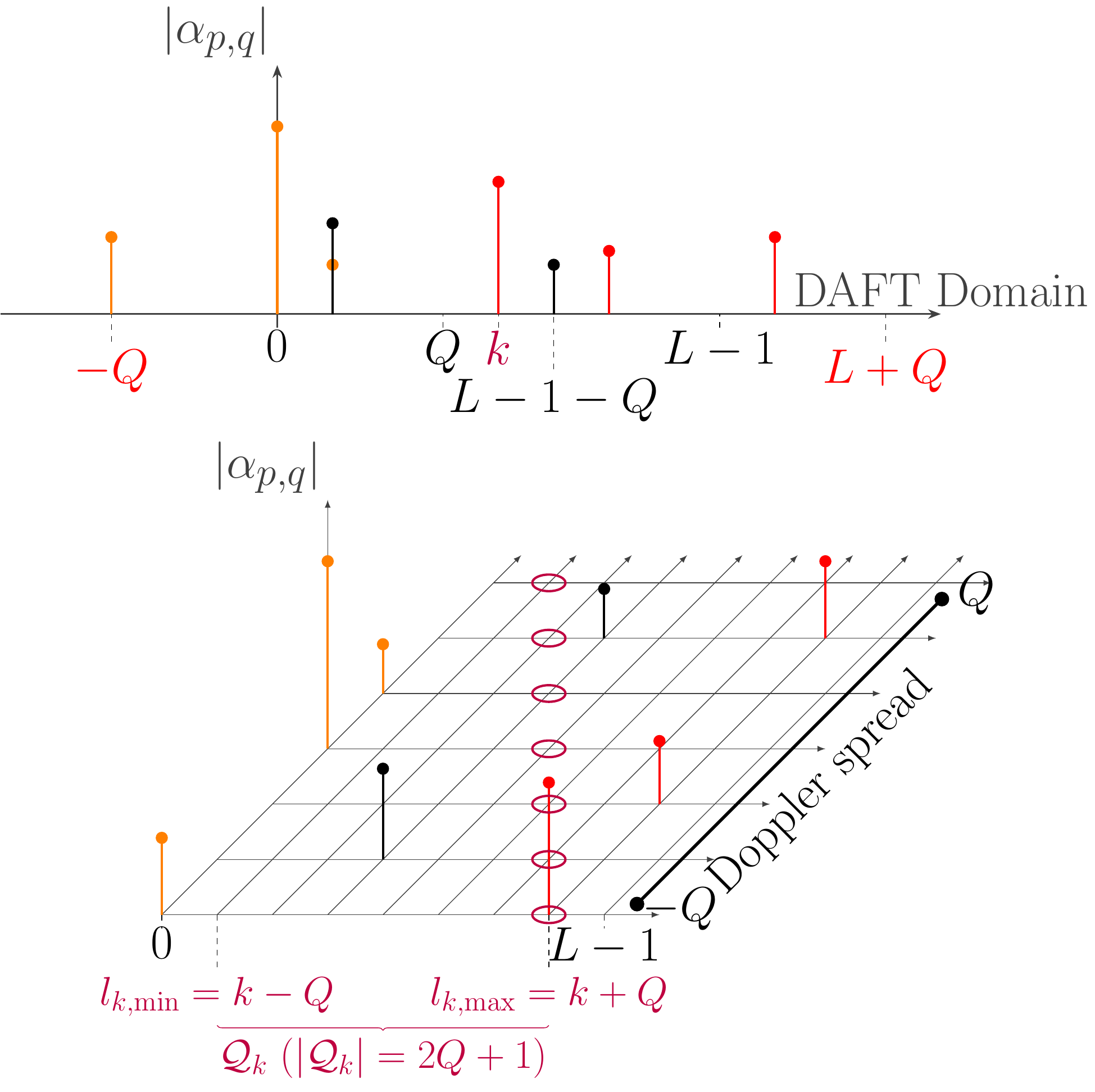} &
    \includegraphics[width=.4\columnwidth]{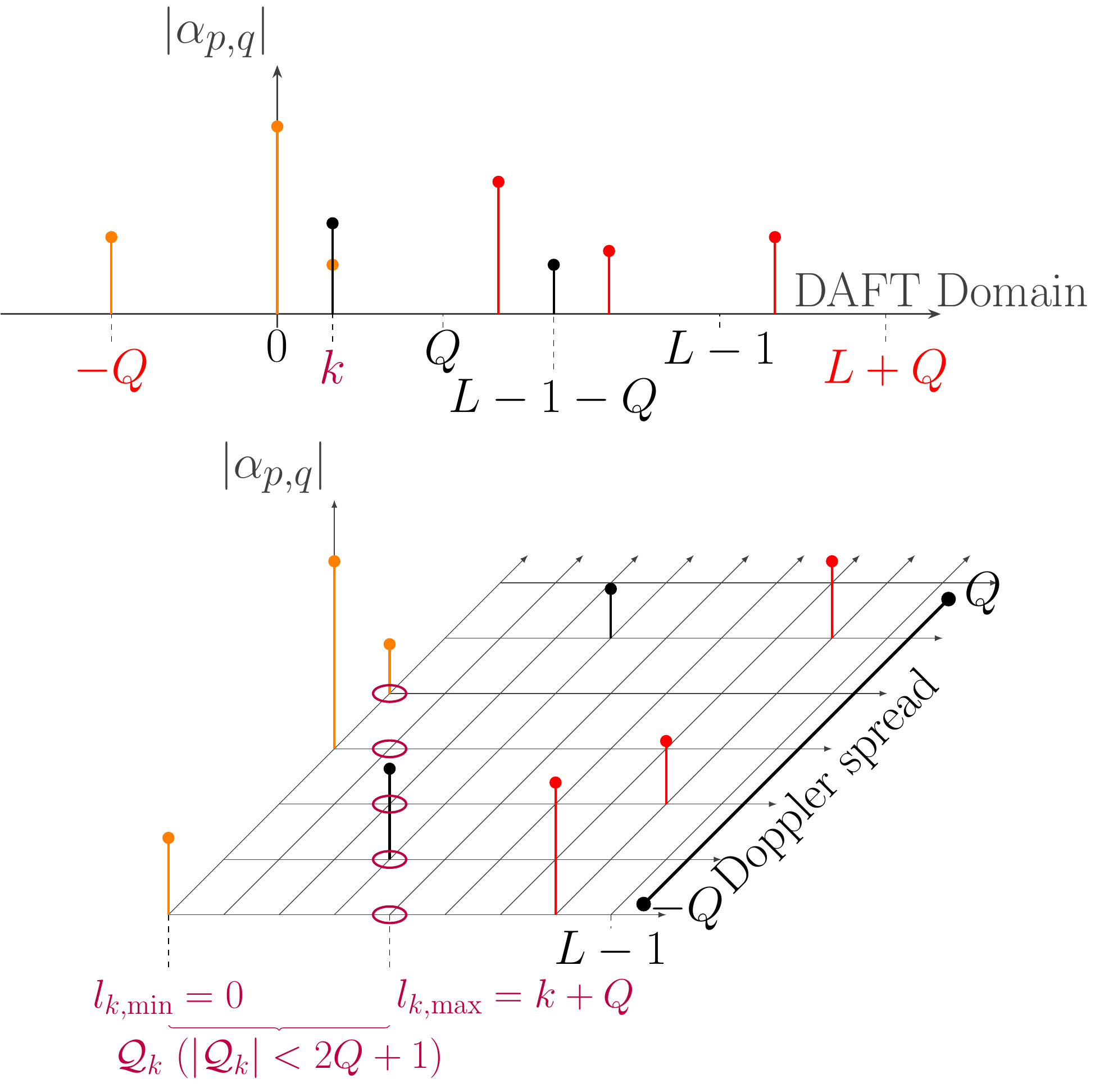} \\
    \small (a) &
      \small (b)
  \end{tabular}
  \medskip
  \caption{Examples of interval $\mathcal{Q}_k$ satisfying, (a) $\rho_k\triangleq\left|\mathcal{Q}_k\right|=2Q+1$, (b) $\rho_k<2Q+1$. Grid points surrounded by circles represent potential delay-Doppler taps that may appear at the $k$-th position in the DAFT domain.}
  \label{fig:interval_k}
  \vspace{-2mm}

\end{figure}
Since $\mathbb{P}[\mathcal{I}_{l,l-k}]=p_{\rm d}p_{\rm D}$ and $\mathbb{P}[\overline{\mathcal{I}}_{l,l-k}]=1-p_{\rm d}p_{\rm D}$ due to Assumption \ref{assum:independence}, we get
\begin{equation}
    \mathbb{P}\left[X_k=M\right]=\left(
\begin{array}{c}
     2Q+1  \\
     M
\end{array}\right)\left(p_{\rm d}p_{\rm D}\right)^{M}\left(1-p_{\rm d}p_{\rm D}\right)^{2Q+1-M}\:.
\end{equation}
Thus, $\forall k\in\Iintv{Q,L-Q},X_k\sim\mathrm{B}(2Q+1,p_{\rm d}p_{\rm D})$. If $F(n,k,p)$ is the cumulative distribution function (CDF) of $\mathrm{B}(n,p)$ then
\begin{equation}
    \label{eq:ccdf1}
    \mathbb{P}\left[X_k>M\right]=1-F(2Q+1,M,p_{\rm d}p_{\rm D}),\quad\forall k\in\Iintv{Q,L-Q}\:.
\end{equation} 
If $k\in\Iintv{-Q,Q-1}\cup\Iintv{L-Q,L-1+Q}$ then $l_{k,\min}=k-Q$ and $l_{k,\max}=k+Q$ cannot be both satisfied and $\left|\mathcal{Q}_k\right|<2Q+1$ as shown in Figure \ref{fig:interval_k}-(b).
 Define $\rho_k\triangleq\left|\mathcal{Q}_k\right|$. Either $\rho_k<M$, in which case $\mathcal{L}_{k,M}=\emptyset$ and $\mathbb{P}\left[X_k=M\right]=0$ or $M\leq\rho_k<2Q+1$, in which case $\left|\mathcal{Q}_k\setminus\{l_1,\ldots,l_M\}\right|=\rho_k-M$ and 
\begin{equation}
    \mathbb{P}\left[X_k=M\right]=\left(
\begin{array}{c}
     \rho_k  \\
     M
\end{array}\right)\left(p_{\rm d}p_{\rm D}\right)^{M}\left(1-p_{\rm d}p_{\rm D}\right)^{\rho_k-M}\:.
\end{equation}
We thus have $X_k\sim\mathrm{B}(\rho_k,p_{\rm d}p_{\rm D})$ leading to
\begin{equation}
    \label{eq:ccdf2}
    \begin{split}
    \mathbb{P}\left[X_k>M\right]=1-F(\rho_k,M,p_{\rm d}p_{\rm D})\\
    \leq1-F(2Q+1,M,p_{\rm d}p_{\rm D})\\
    \forall k\in\Iintv{-Q,Q-1}\cup\Iintv{L-Q,L-1+Q}\:.
    \end{split}
\end{equation}
The inequality in \eqref{eq:ccdf2} follows from the decreasing monotonicity property in $n$ of the CDF of the binomial distribution $\mathrm{B}(n,k,p)$.
Combining \eqref{eq:ccdf1} and \eqref{eq:ccdf2} gives us a uniform upper bound on the CCDF of $X_k$ for any $k\in\Iintv{-Q,L-1+Q}$.
\end{proof}
\section{Sketch of the proof of Theorem \ref{theo:adaptive_M}}
\label{app:proof_theorem_adaptive_M}
\begin{proof}
Set $M=O(1)$. Using Lemma \ref{lem:sub_binomial} and Chernoff's bound applied to $\mathrm{B}\left(2\lceil\frac{Q}{P}\rceil+1,p_dp_D\right)$ it can be shown that $\mathbb{P}[X_k>M]=O(\frac{1}{K})$. 
Next, define $\Delta_{l,q}\triangleq|\alpha_{l,q}-\hat{\alpha}_{l,q}|$, $k_{l,q}\triangleq q-2c_1Nl$ and note that the MSE writes now as $\sum_{l=0}^{L-1}\sum_{q=-Q}^{Q}\mathbb{E}[\Delta_{l,q}^2]$ with
\begin{align}
    \label{eq:total_exp}
    &\mathbb{E}[\Delta_{l,q}^2]=
    \mathbb{E}[\Delta_{l,q}^2|I_{l,q}=0]\mathbb{P}[I_{l,q}=0]+\nonumber\\
    &\mathbb{E}[\Delta_{l,q}^2|I_{l,q}=1,0<X_{k_{l,q}}\leq M]\mathbb{P}[I_{l,q}=1,0<X_{k_{l,q}}\leq M]+\nonumber\\
    &\mathbb{E}[\Delta_{l,q}^2|I_{l,q}=1,X_{k_{l,q}}>M]\mathbb{P}[I_{l,q}=1,X_{k_{l,q}}>M]\:.
\end{align}
The first term in the right-hand side of \eqref{eq:total_exp} is zero. The second and third terms can be bounded uniformly in $(l,q)$ using Lemma \ref{lem:sub_binomial} and the properties of the MMSE estimator (the two bounds are different since the measurement matrix of the partial problem of MMSE estimating the $X_{k_{l,k}}$ unknowns from $M$ measurements is rank-deficient in the third term) in a way that there exists a constant $C$ such that $\sum_{l=0}^{L-1}\sum_{q=-Q}^{Q}\mathbb{E}[\Delta_{l,q}^2]\leq C(\sigma_w^2+\frac{1}{K})$ holds for $K$ large enough $\forall\sigma_w^2$ proving that the MSE tends to zero when the number of pilots $M_{\min}=O(1)$. This number of pilots, each costing $(L-1)P+2Q+1$ samples, results since $P=O(K^{(\kappa_{\rm d}+\kappa_{\rm D}-1)_+})$ in a total overhead $M_{\min}\left((L-1)P+2Q+1\right)=O\left(K^{\kappa_{\rm d}+\kappa_{\rm D}}\right)$.
\end{proof}
\bibliographystyle{IEEEtran}
\bibliography{IEEEabrv,Citations}
\end{document}